\documentclass[a4paper]{article}

\usepackage[utf8]{inputenc}
\usepackage[english]{babel}

\usepackage{fullpage} 

\usepackage[textsize=scriptsize]{todonotes}

\usepackage{authblk}

\usepackage[shortlabels]{enumitem}

\usepackage{amsmath}
\usepackage{amssymb}
\usepackage{bbm} % for indicator 1
\usepackage{mathrsfs}
\usepackage{mathtools}

\usepackage{url}
\usepackage{hyperref}
\usepackage{bookmark}
\hypersetup{
    colorlinks = true,
    citecolor = blue,
    linkcolor = blue,
    bookmarksnumbered
}

\usepackage{graphicx}
\usepackage[width=.9\textwidth]{caption}
\usepackage{subcaption}
\usepackage{float}

\usepackage{tikz}
\usetikzlibrary{calc,fit,shapes}

\usepackage{cite}

\usepackage[ruled, vlined, linesnumbered]{algorithm2e}
\SetKwFor{RepTimes}{repeat}{times}

\usepackage[noabbrev]{cleveref}

\usepackage{xspace}

\usepackage{centernot}

\usepackage[misc]{ifsym}

\usepackage[title]{appendix}

\usepackage{amsthm}
\theoremstyle{plain}
    \newtheorem{theorem}{Theorem}
    \newtheorem{lemma}{Lemma}

\theoremstyle{definition}

\crefname{definition}{definition}{definitions}
\crefname{theorem}{theorem}{theorems}
\crefname{corollary}{corollary}{corollaries}
\crefname{lemma}{lemma}{lemmas}
\crefname{proposition}{proposition}{propositions}
\crefname{claim}{claim}{claims}
\crefname{remark}{remark}{remarks}
\newcommand{\R}{\mathbb{R}}
\newcommand{\Z}{\mathbb{Z}}

\newcommand{\cC}{\mathcal{C}}

\newcommand{\cP}{\mathcal{P}}

\newcommand{\cG}{\mathcal{G}}

\newcommand{\Par}[1]{\left( #1 \right)}

\newcommand{\CBra}[1]{\left\{ #1 \right\}}

\def\keywordname{{\bf Keywords:}}
\providecommand{\keywords}[1]{\def\and{{\textperiodcentered} }
\par\addvspace\baselineskip
\noindent\keywordname\enspace\ignorespaces#1}

\begin{document}

\title{A Note on the Core of 2-Matching Games}
\author[1]{Laura Sanità}
\author[2]{Lucy Verberk\(^{(\text{\Letter})}\)}
\affil[1]{
    Bocconi University of Milan, Italy. \protect \\
    {\normalsize\tt{laura.sanita@unibocconi.it}}
}
\affil[2]{
    Eindhoven University of Technology, the Netherlands. \protect \\
    {\normalsize\tt{l.p.a.verberk@tue.nl}}
}
\date{}
\maketitle

\begin{abstract}
    Cooperative \(2\)-matching games are a generalization of cooperative matching games, where the value function is given by maximum-weight \(b\)-matchings, for a vertex capacity vector \(b \leq 2\). We show how to separate over the core of \(2\)-matching games in polynomial time, fixing a small flaw in the literature, and prove the existence of a compact extended formulation for it.

\end{abstract}

\section{Introduction}

In this paper we consider cooperative \(2\)-matching games, which are a generalization of cooperative matching games, introduced by Shapley and Shubik 50 years ago \cite{Shapley1971Assignment}. An instance of cooperative matching games is defined on a pair \((G,w)\), consisting of a graph \(G = (V,E)\) with \(n\) vertices and \(m\) edges, and edge weights \(w \in \R^E_{\geq 0}\). A \emph{(cooperative) matching game} defined on \((G,w)\) is a pair \((N,\nu)\), consisting of a player set \(N = V\) and a value function \(\nu: 2^N \rightarrow \R_{\geq0}\), defined by
\begin{equation*}
    \nu(S) = w(M_S) = \sum_{e\in M_S} w_e
\end{equation*}
for all \(S\subseteq N\), where \(M_S \subseteq E\) is a maximum-weight matching in \(G[S]\), the subgraph of \(G\) induced by \(S\).
Recall that a matching is a subset \(M \subseteq E\) such that each vertex \(v \in V\) is incident with at most one edge from \(M\).

This paper focuses on an important generalization of matching games, which arises when replacing matchings with \(b\)-matchings (see e.g.\ \cite{Biro2016Stable}), where a \(b\)-matching is a subset \(M \subseteq E\) such that each vertex \(v \in V\) is incident with at most \(b_v\) edges from \(M\).
An instance of cooperative \(b\)-matching games is defined on a triple \((G,b,w)\), which in addition to the graph \(G\) and edge weights \(w\), also specifies vertex capacities \(b \in \Z^V_{\geq0}\). A \emph{(cooperative) \(b\)-matching game} defined on \((G,b,w)\) is again a pair \((N,\nu)\), but now the value function \(\nu\) is defined using maximum-weight \(b\)-matchings.
We refer to \(b\)-matchings games with \(b \leq 2\) as \emph{\(2\)-matching games}.

A fundamental concept in cooperative games is that of \emph{core}. For an arbitrary cooperative game \((N,\nu)\), the core consists of all vectors \(p \in \R^N\), called \emph{allocations}, that satisfy
\begin{equation}
\label{eq: core constraints}
\begin{split}
    & p(N) = \nu(N), \\
    & p(S) \geq \nu(S) \text{ for all } S \subset N,
\end{split}
\end{equation}
where, for any \(S \subseteq N\), \(p(S) = \sum_{i \in S} p_i\). Assuming that a group of players \(S \subseteq N\) forms a coalition, they can distribute a total value of \(\nu(S)\) among them. The equation (\(p(N) = \nu(N)\)) makes sure that the grand coalition \(N\) is formed, and the inequalities make sure that no subset of players has an incentive to leave the grand coalition to form a coalition on their own.

We are interested in the problem of separating over the core:
\begin{center}
    \emph{Determine if a given allocation \(p\) belongs to the core, or find a coalition that violates the corresponding constraint in \cref{eq: core constraints}.}
\end{center}

Separating over the core of matching games is solvable in linear time; given \(p \in \R^N_{\geq0}\) with \(p(N) = \nu(N)\), it is equivalent to checking if \(p_i + p_j \geq w_{ij}\) for all edges \(ij \in E\).
In fact, the core admits a compact linear programming (LP) formulation (given by the dual of the standard LP relaxation of maximum-weight matching).
This was first shown for bipartite graphs by Shapley and Shubik~\cite{Shapley1971Assignment}, and later generalized to arbitrary graphs by e.g.\ \cite{Deng1999Algorithmic,Paulusma2001Complexity}.
Differently, Biró et al.~\cite{Biro2016Stable} show that separating over the core of \(b\)-matching games (which they call multiple partners matching games) is co-NP-complete, even on bipartite graphs with \(b = 3\) and \(w = 1\) \cite[theorem 13]{Biro2016Stable}. On the other hand, \(2\)-matching games seem to still behave nicely: they state that separating over the core of \(2\)-matching games is solvable in polynomial time \cite[theorem 12]{Biro2016Stable}. 
However, their proof contains a small flaw.

\paragraph{Our results}
Our first result is to fix the flaw in the proof of \cite[theorem 12]{Biro2016Stable}, hence showing
\begin{theorem}
    \label{thm: separating}
    Separating over the core of \(2\)-matching games is solvable in polynomial time.
\end{theorem}
Having a polynomial-time separation oracle over the (convex) set of core allocations, implies that we can optimize over the corresponding polytope in polynomial time via the ellipsoid method \cite{Khachiyan1979Polynomial,Grötschel1981Ellipsoid,grotschel2012geometric}. A natural question is then whether there exists a compact extended formulation for it. In fact, there exist polytopes for which a polynomial-time separation oracle is known, but no compact extended formulation exists, such as the perfect matching polytope \cite{Rothvoss2017matching}. Our second result is a positive answer to this question.
\begin{theorem}
    \label{thm: ext form}
    There exists a compact extended formulation that describes the core of \(2\)-matching games.
\end{theorem}
\section{Separating over the Core}\label{sec: separting core}

The first important observation in \cite{Biro2016Stable} is that for any \(S \subseteq N\), a maximum weight \(b\)-matching in \(G[S]\) is composed of cycles and paths, which means the core of 2-matching games can alternatively be described by the following (smaller) set of constraints:
\begin{subequations}
\begin{align}
    \label{eq: core total value}
    & p(N) = \nu(N), \\
    \label{eq: core cycles}
    & p(C) \geq w(C), \quad \text{for all cycles } C \in \cC, \\
    \label{eq: core paths}
    & p(P) \geq w(P), \quad \text{for all paths } P \in \cP.
\end{align}
\end{subequations}
Here, \(\cC\) stands for the set of cycles \(C \subseteq E\) in \(G\) with \(b_i = 2\) for all \(i \in V(C)\), and \(\cP\) stands for the set of paths \(P \subseteq E\) with \(b_i = 2\) for all inner vertices on \(P\). We shortened \(p(V(C))\) and \(p(V(P))\) to \(p(C)\) and \(p(P)\), respectively.
With this observation, separating over the core for a given vector \(p\) reduces to checking whether \(p(N) = \nu(N)\), which can be done in polynomial time (a maximum-weight \(b\)-matching in a graph \(G\) can be computed in polynomial time, see e.g.~\cite{Letchford2008Odd}), and to separating over the set of constraints for cycles, and the set of constraints for paths.

Biro et al.~\cite{Biro2016Stable} show how to separate over the set of cycle constrains, by reducing the problem to the \emph{tramp steamer} problem (also known as the minimum cost-to-time ratio problem), which we will introduce now.
Let \(G = (V,E)\) be a graph with edge weights \(p', w \in \R^E_{\geq0}\). The tramp steamer problem is to find a cycle \(C \subseteq E\) of \(G\) that maximizes the ratio \(w(C) / p'(C)\). The tramp steamer problem is well-known to be solvable in polynomial time (see e.g.\ \cite{Dantzig1966Finding, Lawler1976combinatorial, Eiselt2000Shortest}).

The following lemma is proved in \cite{Biro2016Stable} (see the proof of their Theorem 12). We report a proof for completeness.
\begin{lemma}[\cite{Biro2016Stable}]
    \label{lem: separting cycles}
    Separating over the constraints for cycles in \cref{eq: core cycles} is solvable in polynomial time.
\end{lemma}
\begin{proof}
    Let \(N_2 = \CBra{i\in N : b_i = 2}\) and \(G_2 = G[N_2]\). In \(G_2 = (N_2,E_2)\) we transfer the given allocations \(p_i\) to the edges by setting \(p'_{ij}=\Par{p_i+p_j}/2\) for all \(ij \in E_2\). This defines edge weighs \(p' \in \R^{E_2}\) such that the core constraints for cycles are equivalent to 
    \begin{equation}
        \label{eq: tramp steamer constraint}
        \max_{C\in\cC} \frac{w(C)}{p'(C)} \leq 1.
    \end{equation}
    Hence we obtained an instance of the tramp steamer problem, which is polynomial-time solvable as mentioned before. Note that by solving the above maximization problem we either find that all the constraints for cycles in \cref{eq: core cycles} are satisfied or we end up with a particular cycle \(C\) with \(p(C) = p'(C) < w(C)\).
\end{proof}

Next, we discuss the flaw related to the separation of the path constraints in \cref{eq: core paths}.

\paragraph{Path separation of \cite{Biro2016Stable}}
    Assuming that all the constraints for cycles in \eqref{eq: core cycles} are satisfied by the given vector \(p \in \mathbb{R}^N\), they process the path constraints separately for all possible endpoints \(i_0, j_0 \in N\) (with \(i_0 \neq j_0\)) and all possible lengths \(k = 1, \ldots, n-1\). Let \(\mathcal{P}_k(i_0, j_0) \subseteq \mathcal{P}\) denote the set of \(i_0-j_0\)-paths of length \(k\) in \(G\). They construct an auxiliary graph \(G_k(i_0, j_0)\), that is a subgraph of \(G \times P_{k+1}\), the product of \(G\) with a path of length \(k\). To this end, let \(N_2^{(1)}, \ldots, N_2^{(k-1)}\) be \(k-1\) copies of \(N_2\). The vertex set of \(G_k(i_0, j_0)\) is then \(\{i_0, j_0\} \cup N_1^{(1)} \cup \cdots \cup N_2^{(k-1)}\). Denote the copy of \(i \in N_2\) in \(N_2^{(r)}\) by \(i^{(r)}\). The edges of \(G_k(i_0, j_0)\) and their weights \(\overline{w}\) are defined as
    \begin{align*}
        & i_0 j^{(1)} && \text{for } i_0 j \in E && \text{with weight } \overline{w}_{i_0 j} := p_{i_0} + p_j/2 - w_{i_0 j}, \\
        & i^{(r-1)} j^{(r)} && \text{for } i j \in E && \text{with weight } \overline{w}_{i j} := (p_i + p_j)/2 - w_{i j}, \\
        & i^{(k-1)} j_0 && \text{for } i j_0 \in E && \text{with weight } \overline{w}_{i j_0} := p_i/2 + p_{j_0} - w_{i j_0}. 
    \end{align*}
    They claim that \(p(P) \geq w(P)\) holds for all \(P \in \mathcal{P}\) if and only if the shortest \(i_0-j_0\)-path in \(G_k(i_0, j_0)\) (w.r.t.\ \(\overline{w}\)) has weight \(\geq 0\) for all \(i_0 \neq j_0\) and \(k = 1, \ldots, n-1\).
    
    However, the next example shows that this claim is not true.
    Consider the graph in \cref{fig:counter example}.
    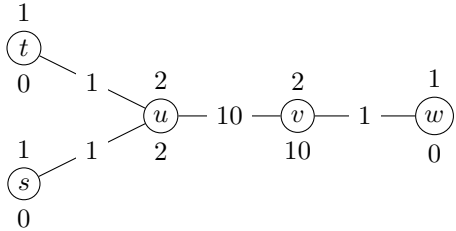
\begin{figure}[ht]
        \centering
        \begin{tikzpicture}[
            scale=.6,
            circ/.style={
                circle,
                draw=black,
                inner sep=2pt
            },
        ]
            \node[circ,label={above:\(1\)},label={below:\(0\)}] (s) at (0,-1.5) {\(s\)};
            \node[circ,label={above:\(1\)},label={below:\(0\)}] (t) at (0,1.5) {\(t\)};
            \node[circ,label={above:\(2\)},label={below:\(2\)}] (u) at (3,0) {\(u\)};
            \node[circ,label={above:\(2\)},label={below:\(10\)}] (v) at (6,0) {\(v\)};
            \node[circ,label={above:\(1\)},label={below:\(0\)}] (w) at (9,0) {\(w\)};

            \draw (s) to node[fill=white] {1} (u);
            \draw (t) to node[fill=white] {1} (u);
            \draw (u) to node[fill=white] {10} (v);
            \draw (v) to node[fill=white] {1} (w);
        \end{tikzpicture}
        \caption{Graph with edge weights \(w\) on the edges, vertex labels in the vertices, vertex capacities \(b\) above the vertices, and an allocation \(p\) below the vertices.}
        \label{fig:counter example}
    \end{figure}
    One can check that \(\nu(N) = 12\) and that the given allocation \(p\) is in the core. In the auxiliary graph \(G_4(s,t)\), defined above, there is a path of total weight strictly less than zero: \(P=(s,u^{(1)},v^{(2)},u^{(3)},t)\) has weight
    \begin{equation*}
        \overline{w}(P) = p_s + 2 p_u + p_v + p_t - w_{su} - 2 w_{uv} - w_{ut} = 14 - 22 < 0.
    \end{equation*}
    According to their claim, this should mean that \(p(P) \geq w(P)\) does not hold for all \(P \in \mathcal{P}\), i.e., \(p\) is not a core allocation.

\paragraph{Our path separation} 
    We now show how to fix the above issue, again by relying on the tramp steamer problem. 
    \begin{lemma}
        \label{lem: separting paths}
        Separating over the constraints for paths in \cref{eq: core paths} is solvable in polynomial time.
    \end{lemma}
    \begin{proof}
        We assume that all the constraints for cycles in \cref{eq: core cycles} are satisfied by the given vector \(p\in\R^N\). We first process all paths of length zero and one separately, i.e., the vertices and edges, by checking
        \begin{equation*}
            \begin{split}
                &p_i \geq 0, \quad \text{for all } i \in N, \\
                &p_i + p_j \geq w_{ij}, \quad \text{for all } ij \in E.
            \end{split}
        \end{equation*}
        We process the remaining constraints for paths separately for all possible endpoints \(s\neq t \in N\). We create an auxiliary graph \(G(s,t)\) as the subgraph of \(G\) induced by the vertex set \(N_2 \cup \CBra{s,t}\). We add the edge \(st\) to \(G(s,t)\) with \(w_{st} = 0\), replacing the original edge \(st\) if it exists. Like we did for the constraints for cycles, we transfer the given allocations \(p_i\) to the edges by setting \(p'_{ij} = \Par{p_i + p_j}/2\) for all edges \(ij\) in \(G(s,t)\). If \(b_s = b_t = 2\), then we solve the tramp steamer problem on \(G(s,t)\) directly. If \(b_s = 1\) and \(b_t = 2\), then we solve \(d_s-1\) (where \(d_s\) is the degree of \(s\) in \(G(s,t)\)) instances of the tramp steamer problem: we remove all edges incident to \(s\) in \(G(s,t)\) except for \(st\) and one other edge \(e\), and solve the tramp steamer problem on this variant of \(G(s,t)\). We repeat this for all possible edges \(e\) incident to \(s\), unequal to \(st\). The case \(b_s = 2\) and \(b_t = 1\) can be handled similarly. If \(b_s = b_t = 1\), then we solve \((d_s-1)(d_t-1)\) instances of the tramp steamer problem: we remove all edges incident to both \(s\) and \(t\) except \(st\), one other edge \(e\) incident to \(s\) and one other edge \(f\) incident to \(t\), and solve the tramp steamer problem on this variant of \(G(s,t)\). We repeat this for all possible combinations of the edges \(e\) and \(f\).
    
        Suppose there is a path \(P\in\cP\) of length at least two, such that \(p(P) < w(P)\). Let \(P=(s;e_1,\ldots,e_k;t)\), \(k\geq2\). There is a variant of \(G(s,t)\) which contains both \(e_1\) and \(e_k\). We construct a cycle \(C\) in this graph from \(P\): \(C=(s;e_1,\ldots,e_k,ts;s)\). Then:
        \begin{equation*}
            \begin{split}
                & w(C) = w_{e_1} + \cdots + w_{e_k} + w_{st} = w(P) + 0 = w(P), \\
                & p'(C) = p'_{e_1} + \cdots + p'_{e_k} + p'_{st} = p(P),
            \end{split}
        \end{equation*}
        which means we have \(w(C)/p'(C) = w(P)/p(P) > 1\). So, solving the tramp steamer problem on this graph, we will find that the maximum is \(>1\).
    
        Suppose we solve the tramp steamer problem on some variant of \(G(s,t)\), for some \(s\) and \(t\), and find that the maximum is \(>1\), and that this is attained by the cycle \(C\). 
        It is straightforward to check that we must have \(s, t \in V(C)\), and in particular \(st \in C\), as otherwise \(b_i = 2\) for all \(i \in V(C)\), which means \(C\in\cC\), contradicting that all the constraints for cycles in \cref{eq: core cycles} are satisfied.
        Let \(P\) be the \(st\)-path obtained from \(C\) by removing \(st\). Then:
        \begin{equation*}
            \begin{split}
                & w(P) = w(C) - w_{st} = w(C) - 0 = w(C), \\
                & p(P) = p'(C),
            \end{split}
        \end{equation*}
        which means we have \(w(P)/p(P) = w(C)/p'(C) > 1\). So the constraint for \(P\) is violated.
    \end{proof}
    
    Note that \cref{lem: separting cycles} and \cref{lem: separting paths}, together with checking \eqref{eq: core total value}, yield a proof of \cref{thm: separating}.
\section{Compact Extended Formulation}

To give a compact extended formulation of the core, we essentially need to rewrite the inequality in \cref{eq: tramp steamer constraint} in a compact form. 

Suppose we are given a graph \(G' = (N',E')\), with edge weights \(p', w \in \R^{E'}_{\geq0}\), and that we want to check whether this graph satisfies the inequality in \cref{eq: tramp steamer constraint}. As a first step, define the edge costs \(c_{ij} = p'_{ij} - w_{ij}\) for all edges \(ij \in E'\). Note that \cref{eq: tramp steamer constraint} is violated if and only if \(G'\) contains a negative cost cycle \(C\) with respect to \(c\), since 
\begin{equation*}
    w(C)/p'(C) > 1
    \iff 0 > \sum_{ij \in C} c_{ij} = \sum_{ij \in C} \Par{p'_{ij} - w_{ij}} = p'(C) - w(C).
\end{equation*}

We therefore focus on checking if a graph \(G' = (N',E')\) with edge costs \(c \in \R^{E'}\) contains a negative cost cycle.
% , and we introduce next a (compact) LP that achieves this goal.
There are several efficient ways to detect negative cost cycles. For example one can rely on the notion of potential in undirected graphs with general edge weights, as described in~\cite{Sebo1997Potentials}. A combinatorial algorithm also follows from~\cite{Dudycz2021Optimal} (that actually works for a broader class of generalized matching problems). For our extended formulation, we rely on the LP formulation designed in~\cite{Barahona1993Reducing} for finding negative cost cycles.
Recall that
a cut \(B \subseteq E'\) is a set of edges of the form \(\delta(X)\) for some \(\emptyset \neq X \subset N'\). We need the following theorem.
\begin{theorem}[Seymour \cite{Seymour1979Sum}]
    \label{thm:Seymour}
    The cone generated by the incidence vectors of the cycles of a graph is defined by the system
    \begin{equation*}
    \begin{split}
        & x_e - x(B \setminus e) \leq 0, \text{ for each cut } B, \text{ for every edge } e \in B, \\
        & x \geq 0.
    \end{split}
    \end{equation*}
\end{theorem}

Using the system of constraints from \cref{thm:Seymour}, we can design an LP formulation as in \cite[section 3]{Barahona1993Reducing}, where the authors define the LP below with the goal of minimizing the cost of a cycle. (In \cite{Barahona1993Reducing} they also add the constraint \(\sum_{e \in E'} x_e = 1\), because they are interested in cycles of minimum mean weight, but here this constraint is not needed.)
\begin{equation}
\label{eq:LP_using_cuts}
\begin{split}
    \min \quad 
        & \sum_{e \in E'} c_e x_e \\
    \text{s.t.} \quad
        & x_e - x(B \setminus e) \leq 0, \text{ for each cut } B, \text{ for every edge } e \in B \\
        & x \geq 0
\end{split}
\end{equation}

One observes that \(G'\) contains a negative cost cycle if and only if there exists a solution to this LP whose objective value is negative (indeed, the LP in this case will be unbounded). This is easily seen as if \(C\) is a cycle with negative cost, its characteristic vector \(x^C\) yields an LP solution with negative objective value. Vice versa, if \(x^*\) is a feasible solution for the LP with negative objective value, by \cref{thm:Seymour} \(x^*\) can be expressed as a conic combination of cycles, implying that at least one such cycle must have negative cost.

To make the above LP compact, we rely on flows. Recall that an \(uv\)-cut is a set of edges of the form \(\delta(X)\) where \(X\) contains exactly one of \(u\) and \(v\). For a fixed edge \(\overline{e} = uv\), the system of inequalities consisting of \(x \geq 0\) and the inequalities of \eqref{eq:LP_using_cuts} for \(\overline{e}\), is then equivalent to
\begin{equation}
\label{eq:rewritten_constr}
\begin{split}
    &  x(B) \geq x_{\overline{e}}, \text{ for each } uv \text{-cut } B \text{ in } G' \setminus \overline{e}, \\
    & x \geq 0,
\end{split}
\end{equation}
which tells us that all \(uv\)-cuts in \(G' \setminus \overline{e}\) have capacity (w.r.t.\ \(x\)) at least \(x_{\overline{e}}\). By the max flow min cut theorem of Ford and Fulkerson \cite{FordFulkerson1962Flows}, this is equivalent to the existence of a \(uv\)-flow (w.r.t.\ the capacity function \(x\)) in \(G' \setminus \overline{e}\) of value \(x_{\overline{e}}\). Therefore there exists an \(x\) that satisfies the constraints in \cref{eq:rewritten_constr} if and only if there is there exists a pair \((x,y)\) that satisfies
\begin{equation*}
\begin{split}
    & \sum_{j : ij \in E' \setminus \overline{e}} \Par{y_{ij} - y_{ji}} = \begin{cases}
        0, & \text{ if } i \neq u, i \neq v, \\
        x_{\overline{e}}, & \text{ if } i = u, \\
        -x_{\overline{e}}, & \text{ if } i = v, \\
    \end{cases}
     \quad \text{ for all } i \in N', \\
    & 0 \leq y_{ij}, y_{ji} \leq x_e, \quad \text{ for all } e = ij \in E' \setminus \overline{e}.
\end{split}
\end{equation*}
Finally, we can rewrite the LP in \cref{eq:LP_using_cuts} as
\begin{equation*}
\label{eq:LP_compact}
\begin{split}
    \min \quad
        & \sum_{e \in E'} c_e x_e, \\
    \text{s.t.} \quad
        & \sum_{j : ij \in E' \setminus \overline{e}} \Par{y_{ij}^{\overline{e}} - y_{ji}^{\overline{e}}} = \begin{cases}
            0, & \text{ if } i \neq u, i \neq v, \\
            x_{\overline{e}}, & \text{ if } i = u, \\
            -x_{\overline{e}}, & \text{ if } i = v, \\
        \end{cases}\\
         &\hspace{5.5cm} \text{ for all } i \in N' \text{ and } \overline{e} = uv \in E', \\
        & 0 \leq y_{ij}^{\overline{e}}, y_{ji}^{\overline{e}} \leq x_e, \quad \text{ for all } e = ij \in E' \setminus \overline{e} \text{ and } \overline{e} \in E'.
       % & (x_e \geq 0, \text{ for all } e \in E', \text{ implied by previous constraint})
\end{split}
\end{equation*}
The dual of this LP is
\begin{equation}
\label{eq: flow dual}
\begin{split}
    \max \quad
        & 0 \\
    \text{s.t.} \quad
        & \begin{rcases}
            \gamma_i^{\overline{e}} - \gamma_j^{\overline{e}} - \lambda_{ij}^{\overline{e}} \\
            \gamma_j^{\overline{e}} - \gamma_i^{\overline{e}} - \lambda_{ji}^{\overline{e}}
        \end{rcases}
        \leq 0, \quad \text{ for all } ij \in E' \setminus \overline{e} \text{ and } \overline{e} \in E', \\
        & \gamma_v^{\overline{e}} - \gamma_u^{\overline{e}} + \sum_{ij \in E' \setminus \overline{e}} \Par{\lambda_{ij}^{\overline{e}} + \lambda_{ji}^{\overline{e}}} \leq c_{\overline{e}}, \quad \text{ for all } \overline{e} = uv \in E', \\
        & \lambda_{ij}^{\overline{e}}, \lambda_{ji}^{\overline{e}} \geq 0, \quad \text{ for all } ij \in E' \setminus \overline{e} \text{ and } \overline{e} \in E'.
       % & ( \gamma_i^{\overline{e}} \text{ free, for all } i \in N' \text{ and } \overline{e} \in E' )
\end{split}
\end{equation}

Combining all of the steps, a graph \(G'\) satisfies the constraint in \cref{eq: tramp steamer constraint} if and only if the LP in \cref{eq: flow dual} attains a feasible solution.

We are now ready to state the compact extended formulation. For convenience, let \(\cG\) be the set of graphs consisting of \(G_2\) and all variants of \(G(s,t)\) for all \(s \neq t \in N\). From \cref{sec: separting core}, we know we need to check \(p(N) = \nu(N)\), \(p_i \geq 0\) for all \(i \in N\), \(p_i + p_j \geq w_{ij}\) for all \(ij \in E\), and that all graphs \(G' = (N', E') \in \cG\) satisfy the constraint in \cref{eq: tramp steamer constraint}. So, in total we get
\begin{equation}
\label{eq: extended formulation}
\begin{gathered}
    p(N) = \nu(N), \\
    p_i \geq 0, \quad \text{for all } i \in N, \\
    p_i + p_j \geq w_{ij}, \quad \text{for all } ij \in E, \\
    \begin{drcases}
        & \begin{drcases}
        \gamma_i^{\overline{e}} - \gamma_j^{\overline{e}} - \lambda_{ij}^{\overline{e}} \\
        \gamma_j^{\overline{e}} - \gamma_i^{\overline{e}} - \lambda_{ji}^{\overline{e}}
    \end{drcases}
    \leq 0, \quad \text{ for all } ij \in E' \setminus \overline{e} \text{ and } \overline{e} \in E', \\
    & \gamma_u^{\overline{e}} - \gamma_v^{\overline{e}} + \sum_{ij \in E' \setminus \overline{e}} \lambda_{ij}^{\overline{e}} + \lambda_{ji}^{\overline{e}} \leq (p_u+p_v)/2 - w_{\overline{e}}, \quad \text{ for all } \overline{e} = uv \in E', \\
    & \lambda_{ij}^{\overline{e}}, \lambda_{ji}^{\overline{e}} \geq 0, \quad \text{ for all } ij \in E' \setminus \overline{e} \text{ and } \overline{e} \in E'.
   % & ( \gamma_i^{\overline{e}} \text{ free, for all } i \in N' \text{ and } \overline{e} \in E' )
    \end{drcases}\\
    \hfill \text{ for all } \ G' = (N', E') \in \cG.
\end{gathered}
\end{equation}
The size of this formulation easily follows from the size of \(\cG\). Consider the graph \(G(s,t)\) for some \(s \neq t \in N\). In the worst case, \(b_s = b_t = 1\), which means there are \((d_s-1)(d_t-1)\) variants of \(G(s,t)\). Therefore 
\begin{equation*}
    |\cG| \leq 1 + \sum_{s \neq t \in N} (d_s-1)(d_t-1) = O(n^4).
\end{equation*}
The formulation in \cref{eq: extended formulation} has
\begin{equation*}
    n + |\cG| \cdot (O(m^2) + O(n m)) = O(n^4 m^2) + O(n^5 m)
\end{equation*}
variables, and 
\begin{equation*}
    1 + n + m + |\cG| \cdot (O(m^2) + O(m) + O(m^2)) = O(n^4 m^2)
\end{equation*}
constraints, i.e., it has polynomial size. This proves \cref{thm: ext form}.

\section{Acknowledgements}
    The authors are supported by the NWO VIDI grant VI.Vidi.193.087.
    The authors thank Kanstantsin Pashkovich for pointing out \cite{Barahona1993Reducing}.

\bibliographystyle{plain}
\bibliography{bib}

\end{document}